\documentclass{article}
\usepackage{bm}
\usepackage{verbatim}       

\usepackage[pdftex]{graphicx}
\usepackage{amsthm}
\usepackage{amssymb}
\usepackage{bm}
\usepackage{amstext}
\usepackage{xcolor}
\usepackage[bookmarks,colorlinks=false]{hyperref}
\usepackage{soul}

\expandafter\let\csname equation*\endcsname\relax

\expandafter\let\csname endequation*\endcsname\relax

\usepackage{amsmath}
\usepackage{amsfonts}
\usepackage{mathrsfs}

\newcommand{\p}{\partial}

\newcommand{\dd}{{\rm d}}

\newcounter{mnotecount}[section]

\newcommand{\mnotex}[1]
{\protect{\stepcounter{mnotecount}}$^{\mbox{\footnotesize $\bullet$\themnotecount}}$
\marginpar{
\raggedright\tiny\em
$\!\!\!\!\!\!\,\bullet$\themnotecount: #1} }

\newcommand{\bd}{\begin{definition}}                
\newcommand{\ed}{\end{definition}}                  
\newcommand{\bc}{\begin{corollary}}                 
\newcommand{\ec}{\end{corollary}}                   
\newcommand{\bl}{\begin{lemma}}                     
\newcommand{\el}{\end{lemma}}                       
\newcommand{\bp}{\begin{proposition}}            
\newcommand{\ep}{\end{proposition}}                
\newcommand{\bere}{\begin{remark}}                  
\newcommand{\ere}{\end{remark}}                     

\newcommand{\bt}{\begin{theorem}}
\newcommand{\et}{\end{theorem}}

\newcommand{\bit}{\begin{itemize}}
\newcommand{\eit}{\end{itemize}}
\newtheorem{theorem}{Theorem}
\newtheorem{corollary}[theorem]{Corollary}
\newtheorem{lemma}[theorem]{Lemma}
\newtheorem{proposition}[theorem]{Proposition}
\theoremstyle{definition}
\newtheorem{definition}[theorem]{Definition}
\theoremstyle{remark}
\newtheorem{remark}[theorem]{Remark}

\numberwithin{theorem}{section}

\title{Projective and amplified symmetries in metric-affine theories}
\author{%
	Alfonso Garc\'{\i}a-Parrado$^{\sharp}$%
		\thanks{E-mail: wtbgagoa@gmail.com} {} and
	Ettore Minguzzi$^\flat$%
		\thanks{E-mail: ettore.minguzzi@unifi.it}
\\[2ex]
	{\small $^\sharp$Departamento de Matem\'aticas, Universidad de C\'ordoba,}\\
	{\small Campus de Rabanales, 14071, C\'ordoba, Spain}\\
	{\small $^\flat$ Dipartimento di Matematica e Informatica ``U. Dini'', Universit\`a
	degli Studi di Firenze}\\
	{\small Via S. Marta 3,  I-50139 Firenze, Italy}
}
\date{\today}

\begin{document}
\maketitle

\begin{abstract}
In this paper we extend the projective symmetry of the full metric-affine Einstein-Hilbert theory to
a new symmetry transformation in the space of affine connections called the {\em amplified symmetry}.
We prove that the Lagrangian of the
Standard Model of particle physics is invariant under this new symmetry.
We also show that the gravitational Lagrangian
can be modified so that the amplified symmetry extends to the gravitational sector and hence to the whole action.
The new theory so constructed
is shown to be dynamically equivalent to Einstein-Cartan's though genuinely metric-affine.
\end{abstract}

\section{Introduction}

In the early twenties of the last century Cartan realized that Einstein's relativity  can be generalized to  connections with torsion. Even before the spin of the electron was discovered he inferred that the source of torsion was the density of intrinsic angular momentum of matter \cite{cartan86}. With later contributions by Sciama \cite{sciama64} and Kibble \cite{kibble61} this led to the  development of the so called Einstein-Cartan-Sciama-Kibble theory, according to which torsion is only present inside matter. Indeed, according to the field equations, torsion does not propagate \cite{hehl76} (modified versions in which torsion propagates  exist \cite{hammond02}).

This theory is still regarded  as one of the most elegant viable generalizations of Einstein's gravity.
Its geometric formulation owes much to Trautman \cite{trautman71,trautman71b,trautman73,trautman75,trautman06} and
to his thorough use of differential forms and covariant exterior differentiation, see also \cite{gockeler89}.

Since torsion enters the Raychaudhuri equation through  terms that do not  have a definite sign \cite{luz17,luz20,pasmatsiou17} it can avert the spacetime singularities predicted by singularity theorems.
The possibly repulsive effect of torsion became even more appealing with
 the discovery of the accelerated expansion of the Universe. Since that discovery, the study of gravitational dynamics with torsion  experienced a resurgence as
authors began  exploring new physics, being it  quantum or classical. It was thought that  other generalizations of Einstein's gravity beside Einstein-Cartan's  that might more easily put to experimental test could be possible.

A variety of alternatives were considered in which torsion or non-metricity do not vanish. In addition to Einstein-Cartan theory  we mention teleparallel gravity \cite{cho76,hayashi79,aldrovandi13}, Weyl conformal gravity \cite{mannheim89,romero12,wheeler18} (with trace of non-metricity), the recently studied symmetric teleparallel gravity \cite{nester99,jarv18,jimenez18,jimenez20} (no curvature, no torsion, with non-metricity), Weyl-Cartan theory \cite{moon10} (with trace of non-metricity and torsion), Eddington-Schr\"odinger purely affine theory \cite{poplawski07} (with torsion and no metric) and the most general metric-affine theories in which the connection is not constrained \cite{hehl95,vitagliano11,vitagliano14}.
The last approach would have been the conceptually more satisfactory because the Palatini's variational
principle is more elegantly formulated without constraining torsion or non-metricity a priori.

In general metric-affine theories the dynamical equation  for the connection has as source the
hypermomentum \cite{hehl76}. This quantity splits into the antisymmetric component, i.e.\ the spin
density, the trace component, vanishing due to projective symmetry of the Einstein-Hilbert Lagrangian,
and the shear traceless component (the so-called shear hypermomentum).

In this work we introduce a new symmetry transformation in the space of affine
connections called {\em the amplified symmetry}.
We confirm that the Lagrangian of
the Standard Model of particle physics is invariant under this symmetry transformation and as a consequence we conclude that in this theory the shear hypermomentum vanishes. Moreover,
we extend the amplified symmetry to the gravitational sector by constructing explicitly a gravitational Lagrangian
that has this invariance (the $\mathcal{L}_0$ Lagrangian). We furthermore prove that the theory
combining the Lagrangian $\mathcal{L}_0$ and the Standard Model Lagrangian is dynamically equivalent to the
Einstein-Cartan Theory.
Therefore our new theory reproduces the Standard Model physics but it is truly metric-affine, as we do not impose any
restriction on the metric or the affine connection.

The structure of the paper is as follows:
section \ref{sec:conventions} introduces our notation and conventions,
section \ref{sec:projective-invariance} is a review of the Einstein-Hilbert full metric-affine
theory, with special emphasis on its projective invariance properties  (Remark \ref{viq}).
Although part of the material might be already known, we felt it necessary to include it here in order to
set up the starting point for deriving our new results
so that a reader familiar with this material can jump straight to
section \ref{cmo}.
We analyze the extension of the projective invariance to the matter sector (subsection \ref{subsec:EH-dynamic})
and discuss the dynamics of this theory and some of its specializations, indicating whether the
projective invariance of the full (gravity plus matter) theory is kept or not (subsection
\ref{subsec:hypermomentum}). Section \ref{cmo} is where our new results are presented:
the definition of the amplified symmetry in Eq. \eqref{kad}, a gravitational Lagrangian $\mathcal{L}_\lambda$ (in Eq. \eqref{bta})
invariant under the  amplified symmetry for $\lambda=0$,
the dynamical equivalence of $\mathcal{L}_0$ coupled
with matter to Einstein-Cartan Theory whenever the matter
sector enjoys the amplified symmetry (Theorem \ref{theo:einstein-cartan-amplified})
and finally the result that this is indeed the case for the Standard
Model Lagrangian (section \ref{sec:standard-model}, Theorem \ref{theo:standard-model}).
Some of the computations of this paper have been double-checked with {\em xTerior}
\cite{xTerior}.

\section{Conventions and formalism}
\label{sec:conventions}
We start by describing the conventions and formalism used in the paper.
We shall work in a four dimensional differentiable manifold $M$ where an arbitrary basis
$\{e_a\}$  and its dual co-basis
$\{e^a\}$, $a=1,\dots,4$ are defined. We do not assume that $\{e_a\}$ are orthonormal.
Latin small indices will represent components
of geometric objects in the basis and we will use small Greek letters to denote
spacetime (abstract) indices of tensors.
As is well-known, a covariant derivative $\nabla$ (affine connection)
on $M$ induces the connection 1-form
$\omega^a{}_b$ and the curvature 2-form by means of the relations
\begin{equation}\label{vaf}
\nabla e^a = -\omega^a{}_b\otimes e^b\;,\quad
\mathcal{R}^a{}_{b}:=\dd \omega^a{}_{ b} +\omega^a{}_{ c} \wedge \omega^c{}_{ b}.
\end{equation}
The components of the curvature 2-form are introduced through
$\mathcal{R}^a{}_{ b}=R^a{}_{ bcd} \tfrac{1}{2} e^c \wedge e^d$,
Our conventions for the Ricci tensor and the
scalar curvature are $R_{ab}:=R^c{}_{a c b}$ and $R:=R_{ab} g^{ab}$, respectively.
The torsion 2-form is
\begin{equation} \label{vid}
T^a:=\dd^\nabla e^a= \dd e^a+\omega^a{}_{ b}\wedge e^b,
\end{equation}
where $\dd^\nabla$ is the covariant exterior differential.
The torsion components are introduced through $T^a=T^a{}_{bc}\tfrac{1}{2} e^b \wedge e^c$.
 Suppose $M$ has a (Lorentzian) metric $g$. If $g_{ab} := g(e_a, e_b)$ then the
non-metricity 1-form is
\[
\mathcal{G}_{ab}:=\dd^\nabla g_{ab},
\]
The  shear (traceless) non-metricity is
$\check{\mathcal{G}}_{ab}:=\mathcal{G}_{ab}-\frac{1}{4}  \mathcal{G} {g}_{ab}$ where
$\mathcal{G}=g^{cd} \mathcal{G}_{cd}$ is the Weyl 1-form. Sometimes we will have to display the additional 1-form
index of $\mathcal{G}_{ab}$. Our convention is then
$$
\mathcal{G}_{ab} = \mathcal{G}_{abc}e^c.
$$
Similarly $\mathcal{G}_d$ represents the basis components of
$\mathcal{G}=\mathcal{G}_d e^d$.

  The Bianchi identities are
\begin{equation} \label{vok}
\dd^\nabla T^a=\mathcal{R}^a{}_{ b} \wedge e^b, \quad \dd^\nabla \mathcal{R}^a{}_{ b}=0, \quad \dd^\nabla \mathcal{G}_{ab}=-\mathcal{R}_{ab}-\mathcal{R}_{ba}.
\end{equation}
With $\eta^{(g)}$ or simply with $\eta$ or $ \dd vol$ we denote the volume form induced by the metric. In components
\[
\eta^{(g)}=\eta^{(g)}_{abcd} \tfrac{1}{4!} e^a\wedge e^b\wedge e^c\wedge e^d,
\]
 $\eta^{(g)}_{abcd}=\sqrt{-\det g_{ab}} [abcd]$ with $[abcd]$ totally anti-symmetric symbol with [1234]=1.

Let $\mathcal{L}=\mathcal{L}(g_{ab},\omega^a{}_{b}, e^a)$ be a general Lagrangian 4-form expressed in terms of connection form, metric coefficients and co-basis. The variation of the integral $\int \!\mathcal{L}$ reads
\[
\delta \!\int \! \mathcal{L}=\int \left(\frac{\delta \mathcal{L}}{\delta g_{ab}}  \delta g_{ab}+\frac{\delta \mathcal{L}}{\delta \omega^a{}_{ b}} \wedge
\delta \omega^a{}_{b} + \frac{\delta \mathcal{L}}{\delta e^a} \wedge \delta e^a\right)
\]
The invariance of the integral under diffeomorphisms implies
 \cite{kopczynski90} \cite[Sec.\ 5.2.1]{hehl95}
\begin{equation} \label{ckf}
\frac{\delta \mathcal{L}}{\delta g_{ab}}   \mathcal{G}_{ab c}\zeta^c+\frac{\delta \mathcal{L}}{\delta \omega^a{}_{b}} \wedge i_\zeta \mathcal{R}^a{}_{b}+\frac{\delta \mathcal{L}}{\delta e^a} \wedge i_\zeta T^a  +\zeta^a \dd^\nabla \frac{\delta \mathcal{L}}{\delta e^a}=0
\end{equation}
for every vector field $\zeta$.
Moreover, if $\mathcal{L}$ depends on just  $(g,\nabla)$ the invariance of the integral under general linear transformations of the co-basis results into the equation \cite[Eq.\ I.3]{trautman71}
\begin{equation} \label{toz}
\frac{\delta \mathcal{L}}{\delta e^a} \wedge e^b  =2 \frac{\delta \mathcal{L}}{\delta g_{rb}}g_{ra}+\dd^\nabla  \frac{\delta \mathcal{L}}{\delta \omega^a{}_{ b}}.
\end{equation}

\section{Projective invariance of the Einstein-Hilbert full metric-affine theory}
\label{sec:projective-invariance}

The Einstein-Hilbert Lagrangian form is
\begin{align*}
\mathcal{L}_{EH}&=\tfrac{1}{2} \eta^{(g)}_{abcd} \mathcal{R}^{a}{}_{ s} g^{sb} \wedge  e^c \wedge e^d=\tfrac{1}{4} \eta^{(g)}_{abcd} {R}^{a b}{}_{ st}   e^s \wedge e^t\wedge e^c \wedge e^d=R \dd vol,
\end{align*}
Let us introduce the Einstein tensor $G_{ab}$ through the identity
\[
\frac{\delta \mathcal{L}_{EH} }{\delta g_{ab}}=-G^{ab}\dd vol,
\]
 where the variation is taken with $\omega^a{}_{b}$ and $e^a$ fixed. This implies
\begin{equation} \label{bka}
G_{ab}=R_{(ab)}-\tfrac{1}{2} R g_{ab}.
\end{equation}
Let us introduce the Hilbert 3-form
\[
\tau_d=\frac{1}{2}\frac{\delta \mathcal{L}_{EH}}{\delta e^d}=\tfrac{1}{2} \eta_{abcd}    \mathcal{R}^a{}_{s} g^{sb} \wedge  e^c
\]
which is such  that $\mathcal{L}_{EH}=\tau_d\wedge e^d$. Let us define $\tau_{d }{}^{r}$ through
\begin{equation} \label{woa}
\tau_d=\tfrac{1}{6} \tau_{d }{}^{r} \eta_{r uvz} e^u \wedge e^v \wedge e^z
\end{equation}
so that
\begin{equation} \label{bkb}
\tau_{ab}= \tfrac{1}{2}(R_b{}^r{}_{ar}  +R^r{}_{ bra})-\tfrac{1}{2} R  g_{ab} .
\end{equation}
Notice that $\tau_{ab}$ is not necessarily symmetric. From (\ref{woa}) we get \begin{equation}
\tau_p\wedge e^q=-\tau_p{}^q \dd vol
\end{equation}
 thus a condition equivalent to symmetry is  $\tau^p\wedge e^q=\tau^q\wedge e^p$.

A {\em projective transformation} is a transformation of the connection forms given by
\begin{equation} \label{jad}
\omega^a{}_{b} \to \omega'{}^a{}_{b}=\omega^a{}_{b}+\delta^a{}_b \alpha
\end{equation}
where $\alpha$ is any 1-form. Of course, this can be also viewed as a
transformation of the affine connection $\nabla$ into a new affine connection $\nabla'$.

\begin{remark} \label{viq}
We present some  known facts  that are useful in order to clarify the role of the projective invariance in the
study of metric-affine theories (and in particular in the full metric-affine Einstein-Hilbert theory). For these type of results see \cite{hehl78,giachetta97,dadhich12}.

The following identities are  easily deducible from  (\ref{vaf})
\begin{align*}
R'^d{}_{cab}&=R^d{}_{cab}+2 \delta^d{}_c \p_{[a} \alpha_{b]},\\
R'_{ab}&=R_{ab}+2 \p_{[a} \alpha_{b]},
\end{align*}
and from them
it follows that  $R^d{}_{cab}-\delta^d{}_c R_{[ab]}$, $R_{(ab)}$, $R$, $\mathcal{L}_{EH}$, $G_{ab}$,  are projectively invariant.

Actually, the tensor $\tau_{ab}$  is also  projectively invariant a fact that does not seem to have been previously noticed. Its invariance can be derived by noticing  that $\mathcal{L}_{EH}$ is invariant and so must be its variation with respect to the co-basis. This observation provides also a quick justification for the projective invariance of $G_{ab}$ as this tensor is obtained from the variation of $\mathcal{L}_{EH}$ with respect to $g_{ab}$.

 The shear non-metricity  $\check{\mathcal{G}}_{ab}$ is also prejectively invariant as it is immediate from $\mathcal{G}'_{ab}=\mathcal{G}_{ab}-2\alpha g_{ab}$. From Eq.\ (\ref{vid}) we have $T'^a-T^a=\alpha \wedge e^a$ and from
 \begin{equation} \label{vob}
 \mathcal{G}'=\mathcal{G}-8 \alpha
 \end{equation}
  we get the invariance of
  \begin{equation} \label{boz}
  T^h{}_{ab}-\tfrac{1}{4} \delta^h{}_{[a} \mathcal{G}_{b]} .
  \end{equation}
   Finally, it is interesting to observe that every tensor $Q_a{}^b$ which is projectively invariant has a projectively invariant
covariant divergence $\nabla_b Q_a{}^b{}$, indeed
   \begin{align*} \label{pre}
\nabla'_b Q_a{}^{ b}&=\nabla_b Q_a{}^{b}- Q_r{}^{ b}(\delta^r{}_a \alpha_b) +Q_a{}^{ r} (\delta_r{}^b\alpha_b)=\nabla_b Q_a{}^{ b}.
\end{align*}
\end{remark}

Of course, as observed in \cite{hehl78,bernal17} unparametrized geodesics are also preserved by  (\ref{jad}), as it follows from a well-known result by Weyl according to which two geodesic equations $\ddot x^a+\Gamma^a_{bc} \dot x^b \dot x^c=0$ and $\ddot x^a+\Gamma'{}^a_{bc} \dot x^b \dot x^c=0$ share the same solutions up to parametrization if and only if
\[
\Gamma'{}^a_{(bc)}=\Gamma^a_{(bc)}+\delta^a{}_{(b} \alpha_{c)}
\]
for some 1-form $\alpha$.
This means that the $\nabla$-free fall is not affected by a projective change (still the geodesics for $
\nabla$ and $\nabla^{(g)}$ could differ if they are not projectively related). Similarly, the reading of
a clock over a timelike geodesic is not affected as this is a metric notion hence independent of the
connection. Since $G_{ab}$ is invariant and equal to the stress-energy momentum tensor many   observables
connected to the source are also expected to be projectively invariant.
This really suggests that
the projective symmetry should be regarded as a full symmetry of nature, an idea that
is implemented by demanding the projective invariance of the matter Lagrangian, as discussed next.

\subsection{Extension of the projective invariance to the matter sector in
the full Einstein-Hilbert theory}
\label{subsec:EH-dynamic}
Soon after Hilbert introduced the variational approach to Einstein's equations, Palatini
\cite{palatini19} noticed that the variation with respect to the metric could be more easily accomplished
in two steps,  first by varying with respect to the metric and torsionless connection and then by considering
the variation of the Levi-Civita connection with respect to the metric. This led Einstein to realize that
the  metric and  connection could be regarded as independent variables \cite{einstein25,ferraris82},
through what is now known as Palatini's approach, and that one could even consider connections with
torsion. However, as recalled by Remark \ref{viq},
the projective change of connection forms
leaves invariant the  Einstein-Hilbert action, thus, as a result, the connection remains undetermined,
though in the projective class of the Levi-Civita connection \cite{trautman71,trautman75,dadhich12}.
In fact it is the dynamical equation for the  connection alone that provides this result
\begin{equation} \label{voa}
\frac{\delta \mathcal{L}_{EH}}{\delta \omega^a{}_{b}}=0 \ \Leftrightarrow \ \exists \alpha: \ \omega^a{}_{b}=(\omega^{(g)})^a{}_{b} +\delta^a{}_b \alpha,
\end{equation}
see also the end of Sec.\ \ref{subsec:hypermomentum} for a proof.
In order to fix the connection $\omega^a{}_{b}$ to the Levi-Civita connection form $(\omega^{(g)})^a{}_{b}$ of general
relativity some `gauge fixing' has to be imposed. Einstein observed that the conditions (a) the connection is
torsionless, or (b) the trace of the torsion vanishes, are sufficient, but some other conditions were later proposed,
for instance one can assume that (c) the connection is metric or that (d) it is compatible with the metric volume form
\cite{hehl81}, $\nabla \eta^{g}=\eta_g \otimes \frac{1}{2} \mathcal{G}=0$, where $\nabla$ is the connection that defines $\omega^a{}_{b}$.
One can also recover the Levi-Civita connection dynamically by introducing additional terms into the Lagrangian which are expressed in terms of torsion and non-metricity
\cite{burton98}
or by means of Lagrange multipliers \cite{hehl78}.

In general the action will also have a matter part
\[
S=S_{EH}+S_M=\int \mathcal{L}_{EH}(g,\nabla)+\int \mathcal{L}_M(g,\p g, \nabla,\phi)
\]
thus this scheme of recovering the metric compatibility condition from the variation of the gravitational action can only work if $\mathcal{L}_M$ does not depend on the independently varied connection $\nabla$, but solely on the Levi-Civita connection $\nabla^{g}$, and hence on the metric and its first derivatives.
Whenever $\mathcal{L}_M$ depends on $\nabla$ one can expect a departure from general relativity.

As soon as the projective invariance of the gravitational action is taken into full account one faces
an important problem: the coupling between $\nabla$ and the matter does not necessarily correspond
to a {\em minimal coupling}. To understand why this happens, recall that the
connection is determined by the dy\-na\-mical equation for $\nabla$. To obtain such an equation it is convenient
to express the Einstein-Hilbert
Lagrangian in terms of the metric components, connection form potentials and co-basis $\{e^a\}$, as follows:
$\mathcal{L}_{EH}=\mathcal{L}_{EH}(g_{ab},\omega^a{}_{ b}, e^a)$, and similarly for the matter part $\mathcal{L}_M$.
Then the variation with respect to $\omega^a{}_{b}$ gives the dynamical equation for $\nabla$
\begin{equation} \label{vis}
\frac{\delta \mathcal{L}_{EH}}{\delta \omega^a{}_{b}}+\frac{\delta \mathcal{L}_{M}}{\delta \omega^a{}_{b}}=0.
\end{equation}
However, the projective invariance of the gravitational action implies
\begin{equation} \label{lol}
\frac{\delta \mathcal{L}_{M}}{\delta \omega^a{}_{ b}} \delta^a{}_b=0.
\end{equation}
In order to satisfy this extra equation one can assume  that $\mathcal{L}_M$
depends only on $g$ and its derivatives, hence on $\omega^{(g)}{}^a{}_{ b}$, but not on $\omega^a{}_{ b}$.
(or $\nabla$).
Although
this is certainly a possibility, it renounces  the idea of minimal coupling between $\nabla$ and the
matter,
thus giving to the connection a lesser important role than expected.

Sandberg \cite{sandberg75}  explored a different possibility. One should not try to remove the projective invariance dynamically. The previous problem with the constraint  (\ref{lol})
could be solved by demanding the matter Lagrangian to be also projectively invariant. In other words
{\em the projective symmetry could be understood as a gauge symmetry of the whole theory, not just of the
gravitational part}.  Among the papers that helped to clarify this point we mention
\cite{hehl78,giachetta97,dadhich12,bernal17,bejarano19}. However, Sandberg considered only matter without
spin and ran into some problems.

\subsection{The dynamics of the affine connection in the full Einstein-Hilbert metric-affine theory}
\label{subsec:hypermomentum}

In this subsection we work out the dynamical equation \eqref{vis}
for the full Einstein-Hilbert metric-affine theory
and discuss the behaviour of each of the resulting conditions under the projective transformation.
Let us start by denoting
$\frac{\delta \mathcal{L}}{\delta \omega_{a b}}:=g^{ac}\frac{\delta \mathcal{L}}{\delta \omega^c{}_{b}}$.   It is convenient to introduce the tensor density $S^{mnh} $
  through
\begin{equation} \label{voc}
\frac{\delta \mathcal{L}_M}{\delta \omega_{ m n}}=-S^{mnh} \tfrac{1}{6}\epsilon_{huvz}e^u\wedge e^v\wedge e^z .
\end{equation}
The density $S^{mnh} $ has been called the {\em hypermomentum}
\cite{hehl76b} \cite{hehl95}. Let us define the trace $S^h=g_{mn}S^{mnh}$ and the {\em shear hypermomentum}
\[
\check S^{mnh}:=S^{(mn)h}-\tfrac{1}{4} S^h g^{mn}.
\]
Then the hypermomentum splits into  an  anti-symmetric contribution $S^{[mn]h}$ and a symmetric contribution $S^{(mn)h}$, where the latter splits further into trace (dilation) and shear parts
\[
S^{mnh}=S^{[mn]h}+\check{S}^{mnh}+\tfrac{1}{4} S^h g^{mn}.
\]
The anti-symmetric term is sometimes called the {\em spin density} for it is easy to show that the Dirac Lagrangian contributes to it.

The dynamical equation for the connection (\ref{vis}) splits into anti-symmetric and symmetric parts.
The anti-symmetric part of (\ref{vis}) is
\begin{equation} \label{bdh}
T^h{}_{ab}-\tfrac{1}{4} \delta^h{}_{[a} \mathcal{G}_{b]} =- S_{[ab]}{}^h -\tfrac{1}{2} S_{[ra]}{}^r   \delta^h{}_b+\tfrac{1}{2}S_{[rb]}{}^r \delta^h{}_a+ \check{\mathcal{G}}^{h}{}_{[ab]}
\end{equation}
where the left-hand side is projectively invariant.
When the non-metricity vanishes this equation becomes the standard dynamical equation for the torsion in the Einstein-Cartan theory.
The symmetric part of (\ref{vis})  splits further into the following equations
\begin{align}
S^h&=0,  \label{cos}\\
2 \check S^{mnh}&=\check{\mathcal{G}}^{nr}{}_{r} g^{h m} + \check{\mathcal{G}}^{mr}{}_{r}g^{h n} - \check{\mathcal{G}}^{hnm}-\check{\mathcal{G}}^{hmn},
\end{align}
where the former equation is due to the projective invariance of $S_{EH}$, cf.\ Eq.\ (\ref{lol}). The latter equation is  equivalent to the following equation \cite{hehl76b}
\begin{align} \label{spd}
 \check{\mathcal{G}}^{abc}&=- \check{S}^{bca}- \check{S}^{acb} + \check{S}^{bac} +\tfrac{1}{2} \check{S}^{ c r}{}_{r} g^{a b}.
\end{align}
Substituting back in (\ref{bdh}) we arrive at
\begin{align} \label{bdk}
\begin{split}
g_{h s} T^s{}_{ab}\!-\!\tfrac{1}{4} g_{h [a} \mathcal{G}_{b]}& =- S_{[ab] h} -\!\tfrac{1}{2} S_{[ra]}{}^r   g_{hb}+\tfrac{1}{2}S_{[rb]}{}^r g_{ha} \\
&\quad  - \check{S}_{h[ba]} + \check{S}_{h[ab]} +\tfrac{1}{2}  g_{h [a} \check{S}_{b] r}{}^{  r}
\end{split}
\end{align}
The torsion cannot be completely determined due to projective invariance of the gravitational action and hence  the projective indeterminacy of the connection. Only the projectively invariant combination on the left-hand side is dynamically determined.


We can now write down all the dynamical equations for $\nabla$
in the full metric-affine theory based on the Einstein-Hilbert Lagrangian. These are
Eq.\ (\ref{spd}) and Eq.\ (\ref{bdk}), while Eq.\ (\ref{cos}) is a constraint on the matter Lagrangian
due to consistency with the projective invariance of the gravitational action. It can be accomplished
dynamically by demanding projective invariance of the matter action as well, as observed by Sandberg
\cite{sandberg75} (hence the
full theory keeps the projective invariance).
These equations also clarify that the source of shear non-metricity
is the shear hypermomentum and that
\[
\check{\mathcal{G}}^{abc}=0 \Leftrightarrow  \check{S}^{abc}=0.
\]
By contrast,  torsion receives a contribution from both spin density and shear hypermomentum.
The dynamical equations also show  that both non-metricity and torsion do not propagate and so are only
expected inside matter.

If $\frac{\delta \mathcal{L}_{EH}}{\delta \omega^a{}_{b}}=0$ then by Eq.\ (\ref{vis}) $S^{abc}=0$ which
simplifies the previous equation leading to Eq.\ (\ref{voa}). Indeed, Eq.\ (\ref{spd}) tells us that the
shear non-metricity vanishes while Eqs.\ (\ref{vob}) and (\ref{bdh}) tell us that with a suitable projective
transformation both the full non-metricity and torsion vanish, thus the connection is Levi-Civita up
to a
projective transformation. This entails that
$G_{ab}=\tau_{ab}$ and $G_{ab}$ coincides with
the usual general relativistic Einstein tensor
constructed from the Levi-Civita connection,
due to the symmetries of the Riemann
tensor for the Levi-Civita connection.
Actually, the identity $G_{ab}=\tau_{ab}$ follows also from Eq.\ (\ref{toz}).

Using identity \eqref{toz}, we deduce that
the independent set of dynamical equations of the full Einstein-Hilbert metric-affine theory,
are those for the affine connection (Eqs.\ (\ref{spd}) and \ (\ref{bdk})) together with
\begin{equation} \label{cis}
2 \tau_a+\frac{\delta \mathcal{L}_{M} }{\delta e^a}=0.
\end{equation}

\subsection{Removing the projective invariance in Einstein-Hilbert metric-affine theories}
\label{subsec:remove-projective}
The previous calculations summarize
the dynamics of the affine connection in the full metric-affine theory based
on the Einstein-Hilbert action but also allow us to consider
more specialized related theories where the projective invariance is removed.
These more specialized theories start from more constrained geometries. The
constraints imply that some variations of the connection are not allowed and so
some hypermomentum components are not defined. Similarly, some variational
equations have to be dropped.

To begin with, observe that under a variation of the connection
$\delta \mathcal{G}_{ab}=-2 \delta \omega_{(ab)}$, so
using the basis independent decomposition
\begin{equation} \label{coa}
\delta \omega_{ab}=\delta \omega_{[ab]}+\check{\delta \omega}_{ab}+ \frac{1}{4}g_{ab} \delta \omega
\end{equation}
 we get
\begin{equation} \label{kki}
 \delta \check{\mathcal{G}}_{ab}=-2  \check{\delta\omega}_{ab}, \qquad  \delta \mathcal{G}=-2 \delta \omega.
\end{equation}
On the other hand,
the  variation of the matter Lagrangian
with respect to the connection reads,  using Eq.\ (\ref{voc})
\begin{align*}
\delta \mathcal{L}_M&
=  - \tfrac{1}{6}\epsilon_{huvz}e^u\wedge e^v\wedge e^z \wedge \left(S^{[ab] h}  \delta \omega_{[ab]} + \check S^{abh}   \check{\delta\omega}_{ab}+  \tfrac{1}{4} S^{h}  \delta \omega \right).
\end{align*}
Thus if $\mathcal{G}=0$ then $S^h$ is not defined. Similarly, if $\check{\mathcal{G}}_{ab}=0$ then $\check S_{abc}$ is not defined.
We see that the spin density $S_{[ab]c}$ is always well defined.





Let us now consider what type of conditions might remove the projective invariance.
Since $T^s{}_{sb}-\tfrac{3}{4}  \mathcal{G}_{b}$ is determined by the dynamics,
as  is evident by taking the trace of (\ref{bdk}), and it is projectively invariant, fixing $T^s{}_{sb}$ has the same effect as fixing $\mathcal{G}_b$ and hence removes
the projective invariance due to (\ref{vob}). Still the most natural way of removing
the projective invariance seems to be through the condition $\mathcal{G}=0$ as this
mechanism does not need to invoke any dynamical equation. It should be observed that
the condition $T^s{}_{sb}-\tfrac{3}{4}  \mathcal{G}_{b}=0$ being projectively invariant,
does not remove the projective symmetry.

Bearing the previous considerations in mind we review next
different metric-affine theories already considered in the literature,
indicating when the projective invariance is preserved, and when it is not.

\begin{quote}

{\bf Non-dilational metric-affine theory}\\
We assume that the geometry is constrained by the condition $\mathcal{G}=0$. The projective invariance of the Einstein-Hilbert full metric-affine theory is
removed and the equations are (\ref{cis}) and
\begin{align}
 \check{\mathcal{G}}^{abc}&=- \check{S}^{bca}- \check{S}^{acb} + \check{S}^{bac} +\tfrac{1}{2} \check{S}^{ c r}{}_{r} g^{a b}.\\
 \begin{split}
g_{h s} T^s{}_{ab}& =- S_{[ab] h} -\tfrac{1}{2} S_{[ra]}{}^r   g_{hb}+\tfrac{1}{2}S_{[rb]}{}^r g_{ha} \\
&\quad  - \check{S}_{h[ba]} + \check{S}_{h[ab]} +\tfrac{1}{2}  g_{h [a} \check{S}_{b] r}{}^{  r}.
\end{split}
\end{align}
It does not make sense to speak of the component $S^h$ of the hypermomentum since the constraint $\mathcal{G}=0$ implies $\delta \omega_{ab} g^{ab}=0$. Namely, the variation that one would need to define $S^h$ is not allowed.
The torsion is fully determined by the dynamical equations as in the Einstein-Cartan theory.

{\bf Dilational Einstein-Cartan theory}\\
We assume that the geometry is constrained by the vanishing of the shear non-metricity, $\check{\mathcal{G}}_{ab}=0$.  The projective invariance of the Einstein-Hilbert metric-affine theory is still present. The equations are (\ref{cis}) and
\begin{align}
\begin{split}
g_{h s} T^s{}_{ab}-\tfrac{1}{4} g_{h [a} \mathcal{G}_{b]}& =- S_{[ab] h} -\tfrac{1}{2} S_{[ra]}{}^r   g_{hb}+\tfrac{1}{2}S_{[rb]}{}^r g_{ha}
\end{split}
\end{align}
and we have still the constraint $S^h=0$ to be imposed on the matter Lagrangian. Again it can be dynamically realized by using a projective invariant $\mathcal{L}_M$.  It does not make sense to speak of shear hypermomentum because the variation needed for its definition cannot be accomplished due to the geometric constraint $\check{\mathcal{G}}_{ab}=0$.

{\bf Einstein-Cartan theory}\\
In the Einstein-Cartan theory one assumes vanishing non-metricity from the outset, $\mathcal{G}_{ab}=0$. The projective invariance of the metric-affine theory gets removed. The equations are (\ref{cis}) and
\begin{align}
g_{h s} T^s{}_{ab}& =- S_{[ab] h} -\tfrac{1}{2} S_{[ra]}{}^r   g_{hb}+\tfrac{1}{2}S_{[rb]}{}^r g_{ha}.
\end{align}
The torsion is fully determined and non-propagating. As above, it does not make sense to speak of the symmetric  hypermomentum component $S^{(ab)c}$, and hence of $\check{S}^{abc}$ or $S^h$.

\end{quote}

\subsection{The shear hypermomentum in the full Einstein-Hilbert metric-affine theory}
As reviewed in the two previous subsections,
the fully Einstein-Hilbert metric-affine theory, or its non-dilational version require the
presence of shear hypermomentum in the dynamics. This raises the question as to what physical
or observed properties of matter could be related to shear hypermomentum.
In flat spacetime the study of the irreducible representations of the Poincar\'e group has clarified that
elementary particles are characterized by mass and spin \cite{weinberg95}.
One possible way of accounting for the shear hypermomentum is to
work with representations of the full affine group.
This is the approach that has been taken in \cite{MAG1977,MAG1978,MAG1979}
where the fermionic contribution to matter is assumed to carry infinite
dimensional representations of the general affine group. Under this framework, the
shear hypermomentum has been identified with certain properties of
of hadronic matter.
This is an intriguing possibility but as far as we understand, it certainly goes beyond the Standard
Model of particle physics that we study in this paper.

One could ask a similar question
for the dilation hypermomentum but in this case we note that
the  gravitational action is projectively invariant. The absence of a dilation
hypermomentum signals that the projective invariance can be promoted to a full symmetry of the total
action as suggested by Sandberg \cite{sandberg75}.

\section{A metric-affine theory with amplified symmetry} \label{cmo}
In this section we extend the projective symmetry studied before to an {\em amplified symmetry}.
The idea is to show that under this new amplified symmetry it is possible to recover a theory
whose dynamics is equivalent to the Einstein-Cartan one while being genuinely metric-affine as the non-metricity is not fixed a priori.
Notice that, as previously observed, in the Einstein-Cartan theory it does not make sense to speak of shear hypermomentum of matter as the action variation defining it is not allowed.
As soon as non-vanishing shear non-metricity is allowed, the shear hypermomentum becomes well defined and one faces the problem of its physical identification.   In our theory, where the shear hypermomentum does indeed make sense, this problem is solved through a mechanism that makes it vanish. In this way there is no need to identify it in some property of matter.
This absence of the shear hypermomentum, $\check{S}^{mnh}=0$,
is obtained much in the same way as $S^h=0$ was obtained under the projective symmetry,
that is, thanks to the amplified symmetry of
the full gravity plus matter action.

The amplified transformation is defined by
\begin{equation} \label{kad}
\omega^a{}_{b} \to \omega'{}^a{}_{ b}=\omega^a{}_{ b}+A^a{}_{ b}, \qquad A_{ab}=A_{ba}
\end{equation}
where $A^a{}_b$ is a  matrix-valued 1-form.

In order to accomplish the amplified symmetry on the gravitational sector we need to modify the Lagrangian.  We do it by adding terms proportional to the non-metricity squared.  For a more radical modification that makes the Lagrangian totally independent of the connection see \cite{harada20}.
Notice that most often authors have modified the Lagrangian to remove the projective symmetry rather than to enlarge it  \cite{burton98}.
Let us consider the  expression
\begin{align}
\mathcal{L}_\lambda
&=\tfrac{1}{8} \, \eta_{abcd}  [(1-\lambda) \dd^\nabla g^{-1} \wedge \dd^\nabla g +4 \mathcal{R}]^a{}_{ s} g^{sb} \wedge  e^c  \wedge e^d \nonumber \\
&=\left( R-\tfrac{1-\lambda}{4}\mathcal{G}^{rs}{}_s \mathcal{G}_r{}^{ t}{}_t +\tfrac{1-\lambda}{4} \mathcal{G}^{uvz} \mathcal{G}_{uzv}\right)\dd vol, \label{bta}
\end{align}
which reduces itself to Einstein-Hilbert's for $\lambda=1$.  Our candidate is $\mathcal{L}_0$ which indeed enjoys the amplified symmetry (see Appendix I), hence
\begin{equation} \label{odd}
\frac{\delta \mathcal{L}_0}{\delta \omega_{ (a b)}} =0.
\end{equation}
For $\lambda\ne 0$, the invariance is restricted to the projective form (\ref{jad}). The projective invariance can be confirmed
by inspection since it can be easily checked that in (\ref{bta}) the non-metricity $\mathcal{G}$ can be replaced with the projective invariant tensor $\check{\mathcal{G}}$.  Thus for $\lambda \ne 0$ we get the weaker equation
\[
\frac{\delta \mathcal{L}_\lambda}{\delta \omega^a{}_{  b}} \delta^a{}_b=0.
\]

\subsection{The canonical metric representative}
\label{subsec:metric-representative}
\bp
The condition of metric compatibility defines an affine space on the space of connections.
\ep
\proof
To prove this observe that if $\bar \omega^a{}_{b}$ is a fixed connection and $\omega{}^a{}_{b}=\bar
\omega^a{}_{ b}+B^a{}_{b}$ is an arbitrary connection, then the condition of metric compatibility for
$\omega^a{}_{b}$ reads
\[
0=\dd^{\bar \nabla} g-B_{ab}-B_{ba}
\]
 which is affine in the coordinates $B^a{}_{ b}$.
 \qed

 A change of connection as above $\tilde \omega{}^a{}_{b}=\omega^a{}_{b}+A^a{}_{b} $, with $A_{ab}=A_{ba}$, does not preserve the non-metricity, instead it defines lines on connection space that are transverse to the hyperplane of metric compatibility and hence they define a projection on it. Thus $\tilde \omega$ is metric compatible if
 \[
 (\dd^{\nabla} g)_{ab}=A_{ab}+A_{ba}=2A_{ab}
 \]
 which allows us to determine the metric compatible projection
 \begin{equation} \label{vvp}
 \tilde \omega{}^a_{\ b}=\omega^a_{\ b}+\tfrac{1}{2} g^{ar}(\dd^\nabla g)_{rb} ,
 \end{equation}
 which we call the {\em canonical metric (connection) representative}.
 Its  torsion is
\begin{equation} \label{vog}
\tilde T{}^a=T^a+\tfrac{1}{2} g^{ar}(\dd^\nabla g)_{rb} \wedge e^b= ( T^a{}_{bc} -\mathcal{G}^a{}_{[bc]}) \tfrac{1}{2} e^b\wedge e^c,
\end{equation}
which, as shown in Theorem \ref{viq-ampl} below, is invariant under the amplified symmetry.
All these considerations lead us to the following result
\bp\label{prop:metric-connection}
The amplified symmetry determines an equivalence relation among connections. In each equivalence class we can select
one canonical metric connection $\tilde \omega$ given by (\ref{vvp})
whose (canonical) torsion is given by (\ref{vog}). The canonical
torsion is invariant under the amplified symmetry.\qed
\ep


\subsection{The dynamical equations for $\mathcal{L}_\lambda$ in the presence of ma\-tter}
\label{subsec:dynamical}
Let us study the variational derivatives of $\mathcal{L}_\lambda$ so as to be able to write down the dynamical
equations of the theory in the presence of matter.
The (modified) Hilbert 3-form is
\begin{align}
\tau_d&=\frac{1}{2}\frac{\delta \mathcal{L}_\lambda}{\delta e^d} =-\tfrac{1}{8} \, \eta_{abcd}  [(1-\lambda) (g^{-1}\dd^\nabla \! g) \wedge (g^{-1}\dd^\nabla \! g) -4 \mathcal{R}]^a{}_{s} g^{sb} \wedge  e^c  \label{wob}
\end{align}
so that $\mathcal{L}_\lambda=\tau_d\wedge e^d$. Let us define $\tau_{d }{}^{r}$ as before through Eq.\ (\ref{woa})
then
\begin{equation} \label{vka}
\tau_{ab}= \tfrac{1}{2}(R_b{}^r{}_{ar}   +R^r{}_{ bra})-\tfrac{1}{2} R \, g_{ab} +\tfrac{1-\lambda}{8} g_{bs} \delta^{suv}_{amn} \mathcal{G}^m{}_{tu}\mathcal{G}^{tn}{}_v .
\end{equation}

Let us introduce the symmetric generalized Einstein tensor $G^{mn}$ through
\[
\frac{\delta \mathcal{L}_\lambda}{\delta g_{mn}}=-G^{mn}\dd vol,
\]
 where the variation is taken with $\omega^a_{\ b}$ and $e^a$ fixed, then
\begin{equation} \label{vkb}
\begin{aligned}
 G^{mn}&=R^{(mn)}-\tfrac{1}{2}R g^{mn}\\
 &\quad
 -\tfrac{1-\lambda}{4} \left\{
 2   ({R}^{(am)}{}_{a}{}^n    +{R}^{(an)}{}_{a}{}^{m} ) +\delta^{puv}_{abc}  \mathcal{G}^{a(m}{}_p   g^{n)b}  (T^c_{uv}-\mathcal{G}^{c}{}_{[uv]} )  \right.\\
&\quad  +2 g^{mn} \left( -\tfrac{1}{4}\mathcal{G}^{rs}{}_s \mathcal{G}_r^{\ t}{}_t +\tfrac{1}{4} \mathcal{G}^{uvz} \mathcal{G}_{uzv}\right) \left. +\mathcal{G}_s{}^{(nm)}\mathcal{G}^{sb}{}_b-\mathcal{G}^{sb(m}\mathcal{G}^{n)}{}_{sb}  \right\}
\end{aligned}
\end{equation}

The dynamical equation for the connection takes the form
\[
g^{rm}\frac{\delta \mathcal{L}_\lambda}{\delta \omega^r_{ \ n}}=S^{mnh} \tfrac{1}{6}\epsilon_{huvz}e^u\wedge e^v\wedge e^z .
\]
 Notice that since the variational derivative on the left-hand side is taken with the components $g_{ab}$ fixed, we can also write the left-hand side   as $\frac{\delta \mathcal{L}_\lambda}{\delta \omega_{ m n}}$. We can split this equation into anti-symmetric and symmetric parts.


As far as the anti-symmetric part is concerned, the left-hand side of the equation $\frac{\delta \mathcal{L}_\lambda}{\delta \omega_{ [m n]}}=S^{[mn]h} \frac{1}{6}\epsilon_{huvz}e^u\wedge e^v\wedge e^z $ can be rewritten as
\begin{equation}
\frac{\delta \mathcal{L}_\lambda}{\delta \omega_{ [m n]}}= g^{ma} g^{nb}\eta_{abcd} e^c\wedge [T^d+\tfrac{1}{2} (g^{-1} \dd^\nabla g)^d{}_r \wedge e^r],
\end{equation}

while the equation itself reads (it is the analog of \cite[Eq.\ I.12]{trautman71} in Einstein-Cartan theory)
\begin{equation} \label{bbh}
T^h{}_{ab}-  \mathcal{G}^h{}_{[ab]} =- S_{[ab]}{}^h -\tfrac{1}{2} S_{[ra]}{}^r   \delta^h{}_b+\tfrac{1}{2}S_{[rb]}{}^r \delta^h{}_a.
\end{equation}
It does not depend on $\lambda$ and displays on the left-hand side an amplified symmetry invariant, i.e.\ the components of the canonical torsion.
Notice that Eq.\ (\ref{bbh}) can be used to simplify Eq.\ (\ref{vkb}) where we took care in isolating a term given by the canonical torsion.
By Eq.\ (\ref{bbh}) the spin density vanishes if and only if the canonical torsion vanishes.



Turning now to the symmetric part, the left-hand side of the equation $\frac{\delta \mathcal{L}_\lambda}{\delta \omega_{ (m n)}}=S^{(mn)h} \frac{1}{6}\epsilon_{huvz}e^u\wedge e^v\wedge e^z$ can be rewritten as
\begin{equation*}
\frac{\delta \mathcal{L_\lambda}}{\delta \omega_{ (m n)}} =  \tfrac{\lambda}{4} \, \eta_{abcd}  g^{sb}
(g^{a m} g^{nz}  +  g^{a n} g^{mz} )  (\dd^\nabla g)_{zs} \wedge  e^c  \wedge e^d ,
\end{equation*}
while the equation itself is equivalent to the next formulas
\begin{align}
S^h&=0, \label{vpo0} \\
 \check{S}^{mnh}&=\tfrac{\lambda}{2}(\check{\mathcal{G}}^{nr}{}_{r} g^{h m} + \check{\mathcal{G}}^{mr}{}_{r}g^{h n} - \check{\mathcal{G}}^{hnm}-\check{\mathcal{G}}^{hmn}), \label{vpo1}
\end{align}
where the latter equation  is equivalent to the following equation
\begin{align} \label{vpo}
 \lambda\check{\mathcal{G}}^{abc}&=- \check{S}^{bca}- \check{S}^{acb} + \check{S}^{bac} +\tfrac{1}{2} \check{S}^{ c r}{}_{r} g^{a b}.
\end{align}
Notice that
 $\lambda=0 \Rightarrow S^{(bc)a} =0$.  Stated in another way, the theory with Lagrangian $\mathcal{L}_0$ is only compatible with a matter action having vanishing shear hypermomentum. In a sense it
 {\em predicts } the observed absence of shear hypermomentum. It is now clear that we just  extended the mechanism already present in Sanberg's theory for $S^h$. The equation $S^{(bc)a} =0$ can be automatically satisfied by the matter analog of Eq.\ (\ref{odd}) once the matter Lagrangian is  invariant under the amplified symmetry.
  Whether this
 can be physically realized will be discussed in section
 \ref{sec:standard-model}.

\subsection{Dynamic equivalence of the $\mathcal{L}_0$ theory
with matter and Einstein-Cartan theory}
\label{subsec:ec-equivalence}
In this subsection we are going to show that the dynamical equations just obtained
for the $\mathcal{L}_0$ theory coupled with matter are indeed equivalent to the
Einstein-Cartan theory as long as the amplified symmetry is carried over to the matter sector.
Let us start by proving a generalization of
Remark \ref{viq}
\begin{theorem}\label{viq-ampl}
The canonical torsion
\begin{equation} \label{box}
T^a{}_{bc} -\mathcal{G}^a{}_{[bc]}
\end{equation}
and its trace  $T^s{}_{sc}- \mathcal{G}^a{}_{[ac]}$ are  amplified symmetry invariant.
The  quantities $\mathcal{L}_{\lambda}$,  $G_{ab}$ and $\tau_{ab}$ are all amplified symmetry invariant for $\lambda=0$ and projectively invariant for $\lambda\ne 0$. Moreover, if the equation  $\frac{\delta \mathcal{L}_{\lambda}}{\delta \omega^a{}_{ b}}=0$ holds (for $\lambda=0$ it is equivalent to the vanishing of  the canonical torsion) then $G_{ab}=\tau_{ab}$ and this tensor coincides with the usual general relativistic Einstein tensor constructed from the Levi-Civita connection.
\end{theorem}

\begin{proof}
The amplified (or projective) symmetry invariance of $G_{ab}$ and $\tau_{ab}$ follows from that of $
\mathcal{L}_{\lambda}$ since they are variational derivatives with respect to the metric and co-basis
respectively. We also checked it through direct calculation. The invariance of (\ref{box}) follows
easily from (\ref{vid}). The identity $G_{ab}=\tau_{ab}$ follows from Eq.\ (\ref{toz}) applied to $
\mathcal{L}_\lambda$. If $ \lambda=0$ as Eq.\ (\ref{bbh}) holds, the canonical torsion va\-nish\-es,
which means that the canonical metric representative is the Levi-Civita connection. By the amplified
symmetry,  $G_{ab}$ can be calculated using the Levi-Civita connection, in which case its expression
reduces to the standard one due to the symmetries of the  Riemann curvature for the Levi-Civita
connection. If $\lambda \ne 0$ we have only projective invariance but from Eq.\ (\ref{vpo}) we see that
the shear non-metricity vanishes. Through a projective change, cf.\ Eq.\ (\ref{vob}), we get a
connection for which the whole non-metricity vanishes while $G_{ab}$ remains the same. The Eq.\
(\ref{bbh}) tells us that $
T^h{}_{ab}-  \mathcal{G}^h{}_{[ab]} =0$ where the expression on the left-hand side is projectively
invariant, thus the new representative has also vanishing torsion thus it is the Levi-Civita connection.
 \end{proof}
We are now ready to prove the promised equivalence to the Einstein-Cartan theory.
\begin{theorem}\label{theo:einstein-cartan-amplified}
If $\mathcal{L}_M$ is invariant under the amplified symmetry then,
for $\lambda=0$, the dynamics of the theory described by the
Lagrangian $\mathcal{L}=\mathcal{L}_0 + \mathcal{L}_M$ is equivalent to
the dynamics described by the Einstein-Cartan theory with $\mathcal{L}_M$ as
matter content.
\end{theorem}

\begin{proof}
Observe that for $\lambda=0$, $\tau_{ab}$ and $G_{ab}$ are amplified symmetry invariant
and thus, according to Proposition \ref{prop:metric-connection},
we can choose to evaluate them on the canonical metric representative.
By the additional curvature
symmetry which follows
from Eq.\ (\ref{vok}) the expressions for these tensors (\ref{vka})-(\ref{vkb}) simplify to those for
metric compatible
connections as in the Einstein-Cartan theory (\ref{bka})-(\ref{bkb})
and therefore, we conclude that the dynamical equations that involve these tensors
agree with their Einstein-Cartan theory counterparts.
As for the dynamical equations for the affine connection, Eq.\ \eqref{bbh} establishes a relationship between the torsion of the canonical metric connection and the spin density which is precisely that of the Einstein-Cartan theory. Equations  \eqref{vpo0}-\eqref{vpo1} are satisfied (recall that $\lambda=0$) because the matter Lagrangian is demanded to be amplified symmetry invariant and hence $S^{(ab)c}=0$.
\end{proof}

In conclusion, Theorem \ref{theo:einstein-cartan-amplified} tells us that
the $\mathcal{L}_0$ gravitational theory coupled to matter invariant under the amplified
symmetry becomes equivalent to the Einstein-Cartan theory while being  really metric-affine.

\section{Invariance of the Standard Model under amplified symmetry}
\label{sec:standard-model}
Let us discuss whether the amplified symmetry is satisfied by the matter Lagrangian
of the Standard Model.
So far the basis $\{e_a\}$ was arbitrary, but for discussion of the Dirac equation  it will be
convenient to work with an orthonormal basis, that is $g_{ab}= \eta_{ab}$.
The next result shows that the decomposition (\ref{coa}) of the variation of the connection can be
introduced for the connection itself provided we stay in orthonormal bases (frames).
\begin{lemma}\label{lem:connection}
Let $\nabla$ be a linear connection on $TM\to M$ and let us denote with $\omega^a{}_{ b}$  the
connection 1-form potentials on $M$ in possibly anholonomic bases. Let $g$ be a metric on $M$
(neither compatibility of $\omega^a{}_{ b}$ with the metric nor vanishing torsion are  assumed). Then
for changes between $g$-orthonormal bases, the 1-forms $g^{a m} \omega_{[mb]}$ transform as potentials
for (and hence define) a linear connection which is compatible with the metric while $\omega_{(ab)}$
transforms as a tensor.
\end{lemma}
Thus the theorem tells us that, given a metric, it makes sense to speak of the antisymmetric and
symmetric parts of a connection and it gives a simple way to calculate these parts  in an orthonormal
basis. The components in a non-orthonormal basis are not so easily computed starting from
$\omega^a{}_{ b}$.
\begin{proof}
Under the change of basis
\begin{align*}
\bar e^b&=  G^b{}_a e^a,\\
\bar g_{as}&= g_{ut} (G^{-1})^u{}_a  (G^{-1})^t{}_s,\\
\bar \omega^a{}_{ b}&=  G^a{}_r  \dd (G^{-1})^r{}_b+  G^a{}_r \omega^r{}_{  s}  (G^{-1})^s{}_b .
\end{align*}
Thus
\[
\bar \omega_{ a b}=(G^{-1})^u{}_a   g_{ur}  \dd (G^{-1})^r{}_b+    (G^{-1})^u{}_a  \omega_{u s}  (G^{-1})^s{}_b,
\]
from which we get
\begin{align*}
\bar \omega_{ [a b]}&= (G^{-1})^u{}_a   g_{ur}  \dd (G^{-1})^r{}_b+(G^{-1})^u{}_a  \omega_{[u s]}   (G^{-1})^s{}_b \\
&\quad + \tfrac{1}{2} (G^{-1})^u{}_b   (\dd g_{ur})  \dd (G^{-1})^r{}_a -\tfrac{1}{2}\dd \bar g_{ab},\\
\bar \omega_{ (a b)}&= (G^{-1})^u{}_a  \omega_{(u s)}   (G^{-1})^s{}_b+\tfrac{1}{2} \dd\bar g_{ab}\\
&\quad -\tfrac{1}{2}(G^{-1})^u{}_b   (\dd g_{ur})  \dd (G^{-1})^r{}_a .
\end{align*}
We obtain the desired conclusion using $\dd g_{ur} =\dd \bar g_{ab}=0$. The fact that the antisymmetric part defines a metric connection $\tilde \nabla$ is clear from $(\dd^{\tilde \nabla} g)=\dd g_{ab}-\omega_{[ab]}-\omega_{[ba]}=0$.
\end{proof}

\begin{theorem}\label{theo:standard-model}
The matter Lagrangian of the Standard Model is invariant under the amplified
symmetry.
\end{theorem}

\begin{proof}
In the bundle of orthonormal bases,
the  kinetic term of the Dirac Lagrangian is given by \cite{cho78,hammond02}
\begin{equation}
 \mathcal{L}_T=k \sqrt{-g} [  \tfrac{i}{2} (  \Psi \gamma^a e^\mu{}_a D_\mu \Psi - \overline{D_\mu\Psi} \gamma^a e^\mu{}_a  \Psi)],
\label{eq:dirac-lagrangian}
\end{equation}
 where $k$ is a proportionality constant and where
 \[
 D_\mu = \partial_\mu +\tfrac{1}{2} \omega_{\mu}{}^{ab} \sigma_{ab}, \qquad  \sigma_{ab}=-\tfrac{1}{4} \left[\gamma_{a},\gamma_{b}\right] .
  \]
Notice that we displayed the 1-form index of the connection
according to the convention $\omega^a{}_{b}=\omega_\mu{}^a{}_{b}dx^\mu$.
Clearly $\mathcal{L}_T$ satisfies the amplified
symmetry because $\sigma_{ab}=-\sigma_{ba}$ so that $\omega_{\mu}{}^{ab}$ can actually be replaced by
$\omega_{\mu}{}^{[ab]}$ which, by Lemma \ref{lem:connection}, can be regarded as a connection in the bundle of
orthonormal bases.
Now, in the matter Lagrangian of the Standard Model, the spacetime connection $\omega^a{}_{ b}$ only appears
explicitly in
the kinetic term of each fermionic field, which adopts the form
given by (\ref{eq:dirac-lagrangian}) and therefore the previous argument applies for each of them.
Having confirmed that the part of the Standard Model matter Lagrangian that contains the spacetime connection is
invariant under the amplified symmetry, we conclude that the full matter Lagrangian of the Standard Model
enjoys such invariance.
\end{proof}

From the previous proof we deduce, using a standard argument about the invariance of the
action that, provided the Dirac Lagrangian is generalized as done in this section,
$\frac{\delta \mathcal{L}_{M}}{\delta \omega_{ (m n)}}=0$
which from Eq.\ (\ref{voc}) reads: $S^h=0$ and $\check S^{abc}=0$.
 Therefore, we have demonstrated that, in the context of the geometry used in this work (i.e.\ ignoring infinite dimensional representations of the affine group \cite{MAG1977,MAG1978,MAG1979}), there cannot
be shear hypermomentum in the Standard Model, regardless of the choice of
the gravitational Lagrangian. For this type of result see also \cite{jimenez20,MAG1977,MAG1978,MAG1979}.
Likewise, it extends  the result of Theorem \ref{theo:einstein-cartan-amplified}
to  situations in which $\mathcal{L}_M$ corresponds to the Lagrangian of the
Standard Model. Our predictions agree with the predictions of the Einstein-Cartan
theory as reviewed in subsection \ref{subsec:remove-projective} but the mechanism
employed is totally different:  the $\mathcal{L}_0+\mathcal{L}_M$ theory with
$\mathcal{L}_M$ the Lagrangian of the Standard Model is genuinely metric-affine and leads to a vanishing shear hypermomentum, whereas the Einstein-Cartan theory, by imposing vanishing non-metricity a priory, is not genuinely metric-affine (it does not support  shear hypermomentum as this quantity cannot be defined in this theory). Of course, in both theories, in contrast to metric-affine theories based on different Lagrangians, one is not bothered with the problem of identifying matter with non-vanishing shear hypermomentum.

\section{Conclusions}

The very foundations of differential geometry suggest that the metric and the
connection should be regarded as independent variables. This point of
view, pioneered by Palatini and Einstein, naturally leads to the consideration of  general metric-affine theories
in which no constraint
on the connection is imposed. The Einstein-Cartan theory is not metric-affine since the compatibility with the
metric is demanded  from
the outset. However, we have shown that in the metric-affine theory based on the Einstein-Hilbert action the shear
non-metricity depends
solely on the shear hypermomentum and that the latter vanishes for the Standard Model of particle physics (and
hence for any reasonable
matter Lagrangian) due to the fact that the Standard Model Lagrangian enjoys the {\em amplified symmetry}, which is
a convenient generalization of the projective symmetry.
This means that the Einstein-Cartan dynamics can be obtained from a metric-affine theory based on the Einstein-
Hilbert action.

Nevertheless, it is a bit upsetting that the gravitational sector in such a metric-affine approach could be
consistent with a non-vanishing shear hypermomentum, for it appears to be  just a chance that the matter Lagrangian
does not induce shear hypermomentum. A more satisfactory approach is obtained through a modification of the
gravitational sector such that the gravitational action too becomes amplified symmetry invariant. In this way the
amplified symmetry can be regarded as a true symmetry of nature.
This new metric-affine theory is still equivalent to Einstein-Cartan’s but now it predicts the
vanishing of the shear hypermomentum so that the interpretational issues connected with a potential non-zero value of the last quantity are completely solved.

\section*{Acknowledgements}
The first author thanks the financial support from Grant 14-37086G and the
consecutive Grant 19-01850S of the Czech Science Foundation.
The second author thanks the Institute of Theoretical Physics of Charles University in Prague for kind hospitality.

\section*{Appendix I: Proof of the invariance under amplified symmetry of $\mathcal{L}_0$}
The matrix-valued 2-form
\[
W_{cd}=\eta_{abcd}  [4 \mathcal{R}-g^{-1} \dd^\nabla g  \wedge g^{-1} \dd^\nabla g]^a{}_{ s} g^{sb}
\]
enters the expression of the Lagrangian $\mathcal{L}_0$, indeed $\mathcal{L}_0=\frac{1}{8} W_{cd} e^c \wedge e^d$, cf.\ Eq.\ (\ref{bta}). In this section we check its invariance  under  the (amplified) changes
\begin{equation}
\omega^a{}_{b} \to \omega'{}^a{}_{ b}=\omega^a{}_{ b}+A^a{}_{ b}, \qquad A_{ab}=A_{ba}
\end{equation}
where $A^a{}_b$ is a matrix-valued 1-form. We shall also write $A= A^a{}_b e_a \otimes e^b$. The curvature transforms as follows
\begin{align*}
\mathcal{R}'{}^a{}_s=\mathcal{R}{}^a{}_s+\dd^\nabla A^a{}_s+A^a{}_c \wedge A^c{}_s
\end{align*}
where
\[
\dd^\nabla A^a{}_s=\dd A^a{}_s+\omega^a{}_c \wedge A^c{}_s+A^a{}_c \wedge \omega^c{}_s
\]
is the covariant exterior differential of $A^a{}_s$.
As for the change in the non-metricity we have
\begin{align*}
g^{ar}\dd^{\nabla '} \!\! g_{r b}&=g^{ar} \dd^\nabla \! g_{r b}-A_{b}{}^a-A^a{}_{  b}=g^{ar} \dd^\nabla\! g_{r b}-2A^{a}{}_b.
\end{align*}
 From the expression for $W_{cd}$  we see that the quadratic term in $A$ is cancelled  because an opposite contribution comes from the curvature (no such cancellation can take place for $\mathcal{L}_\lambda$, $\lambda\ne 0$), hence, after suppressing the indices,
 \[
 g^{-1}\dd^{\nabla '}g=g^{-1}\dd^\nabla g-2A
 \]
  and
\begin{align*}
(g^{-1}\dd^{\nabla'} g) \underset{\dot{}}{\wedge} (g^{-1}\dd^{\nabla '}g)& =(g^{-1}\dd^\nabla g) \underset{\dot{}}{\wedge} (g^{-1}\dd^\nabla g)+4 A \underset{\dot{}}{\wedge} A\\
&\quad -2 A \wedge g^{-1}\dd^\nabla g -2 g^{-1}\dd^\nabla g \wedge A.
\end{align*}
Substituting
\begin{align*}
&[4 \mathcal{R}'-(g^{-1}\dd^{\nabla '}g) \underset{\dot{}}{\wedge} (g^{-1}\dd^{\nabla '}g)]g^{-1} \\ &=[4\mathcal{R}-(g^{-1}\dd^\nabla g) \underset{\dot{}}{\wedge} (g^{-1}\dd^\nabla g)]g^{-1}+4 (\dd^\nabla  A)g^{-1}\\
&\quad +2 A \wedge g^{-1} (\dd^\nabla g) g^{-1} +2 g^{-1}\dd^\nabla g \wedge A g^{-1}\\
&=[4\mathcal{R}-(g^{-1}\dd^\nabla g) \underset{\dot{}}{\wedge} (g^{-1}\dd^\nabla g)]g^{-1}+4 (\dd^\nabla  A)g^{-1}\\
&\quad -2 A \wedge  \dd g^{-1}    +2 g^{-1} \dd^\nabla g  \wedge  A g^{-1}\\
&=[4\mathcal{R}-(g^{-1}\dd^\nabla g) \underset{\dot{}}{\wedge} (g^{-1}\dd^\nabla g)]g^{-1}+2 \dd^\nabla  (Ag^{-1})\\
&\quad +2 g^{-1} g  (\dd^\nabla  A)g^{-1}  +2 g^{-1} \dd^\nabla g  \wedge  A g^{-1}\\
&=[4\mathcal{R}-(g^{-1}\dd^\nabla g) \underset{\dot{}}{\wedge} (g^{-1}\dd^\nabla g)]g^{-1}+2 \dd^\nabla  (Ag^{-1})\\
&\quad +2 g^{-1}  \dd^\nabla (g  A)g^{-1}  .
\end{align*}
The last two terms are symmetric and so vanish after contraction with $\eta_{abcd}$. In conclusion, we proved the invariance
\[
W'{}_{cd}=W_{cd},
\]
and hence $\mathcal{L}'_0=\mathcal{L}_0$.


\end{document}